\documentclass[11pt,a4paper]{article}
\usepackage[margin=1in]{geometry}
\usepackage{amsmath,amssymb,amsthm}
\usepackage[hidelinks]{hyperref}
\usepackage{xcolor} 

\newtheorem{theorem}{Theorem}[section]
\newtheorem{proposition}[theorem]{Proposition}

\theoremstyle{remark}
\newtheorem{remark}[theorem]{Remark}
\theoremstyle{definition}

\newcommand{\Def}{\mathrm{Def}}
\newcommand{\Ric}{\mathrm{Ric}}
\newcommand{\divg}{\mathrm{div}}

\newcommand{\R}{\mathbb{R}}
\newcommand{\HH}{\mathbb{H}}
\newcommand{\PP}{\mathbb{P}}
\newcommand{\TT}{\mathbb{T}}

\title{Exponential thermalisation of viscous fluids\\
on negatively curved manifolds}
\author{Samuel L.\ Braunstein${}^{1}$ and Zhi-Wei Wang${}^{2,*}$\\[6pt]
\small ${}^{1}$Computer Science, University of York, York YO10 5GH, UK\\
\small ${}^{2}$College of Physics, Jilin University, Changchun 130012, China\\
\small $^*$Corresponding author e-mail: {zhiweiwang.phy@gmail.com}}
\date{\today}

\begin{document}
\maketitle

\begin{abstract}
The deterministic incompressible Navier-Stokes equations are
physically incomplete: any viscous fluid at finite temperature must
exhibit thermal fluctuations whose form is dictated by the
fluctuation-dissipation relation. We formulate the stochastic
Navier-Stokes equations with the kinematically selected deformation
Laplacian on compact Riemannian manifolds with strictly negative
Ricci curvature. The fluctuation-dissipation relation, derived from a
topological (Poincar\'e lemma) argument, uniquely determines the
noise from the viscous operator. For the spectrally truncated system,
we prove that the unique stationary distribution is the Gibbs measure
(Gaussian in the mode amplitudes, because the nonlinear convective
terms preserve energy), and that convergence to equilibrium is
exponentially fast with rate at least $2\nu\lambda_\Def$, where
$\nu$ is the kinematic viscosity and $\lambda_\Def$ is the spectral
gap of the deformation Laplacian. The spectral gap satisfies
$\lambda_\Def \geq \kappa^2$ when $\Ric \leq -\kappa^2 g$,
and is independent of the volume of the domain. On flat space, the
analogous thermalisation rate vanishes in the infinite-volume limit.
The equilibrium velocity-velocity correlation function decays
exponentially in geodesic distance, in contrast to the algebraic
decay on flat space. These results provide a rigorous statistical-mechanical
foundation for viscous fluids on negatively curved manifolds and
illustrate how the geometry of the domain controls not only the
deterministic dynamics but also the approach to thermal equilibrium.
\end{abstract}

\section{Introduction}
\label{sec:intro}

The incompressible Navier-Stokes equations describe the motion of a
viscous fluid in terms of a deterministic velocity field. Yet any real
fluid at finite temperature exhibits thermal fluctuations: random
perturbations driven by the molecular degrees of freedom that are
coarse-grained away in the continuum description. The existence and
form of these fluctuations are not optional: they are required by the
fluctuation-dissipation theorem, which relates the viscous dissipation
to the amplitude of the thermal noise through the requirement that
the system admit a Boltzmann equilibrium.

In earlier work~\cite{Braun2010}, one of us showed that the
fluctuation-dissipation relation for the Navier-Stokes equations can
be derived from purely topological considerations: Poincar\'e's lemma
on a contractible phase-space domain converts the stationarity
condition for the Fokker-Planck equation into an algebraic relation
between the non-Hamiltonian (viscous) drift and the diffusion matrix.
The resulting stochastic Navier-Stokes equations have a noise term
whose amplitude is proportional to $\sqrt{k_BT\nu k^2}$ in Fourier
space, with a non-trivial momentum dependence and an
incompressibility-preserving projector.

The noise prescription for incompressible fluids was claimed by
Forster, Nelson, and Stephen \cite{FNS77} and is widely assumed
in the stochastic hydrodynamics literature. For compressible
fluids, Zubarev and Morozov~\cite{ZM83} derived the
fluctuation-dissipation relation with multiplicative noise from
microscopic dynamics. The topological derivation
in~\cite{Braun2010} provides an independent, macroscopic route
that requires no microscopic model.

In a companion paper~\cite{WB2026}, building on the identification
by Ebin and Marsden~\cite{EM70} of the Lie derivative
$\mathcal{L}_u g$ as the deformation tensor for viscous fluids on
manifolds, we established that the kinematically correct viscous
operator on a Riemannian manifold is the deformation Laplacian
$\Delta_\Def = \Delta_B + \Ric$, not the Hodge or Bochner
Laplacian. On manifolds with strictly negative Ricci
curvature ($\Ric \leq -\kappa^2 g$), the deformation Laplacian has a
spectral gap: $\langle -\Delta_\Def u, u\rangle \geq
\kappa^2\|u\|_{L^2}^2$ for divergence-free $u$.

The present paper combines these two results. We formulate the
stochastic Navier-Stokes equations with the deformation Laplacian on a
compact negatively curved manifold, derive the fluctuation-dissipation
noise from the topological theorem, and prove that the spectrally
truncated system thermalises exponentially fast with a rate determined
by the spectral gap. The rate is independent of the volume of the
domain, in contrast to the flat-space case where the thermalisation
rate vanishes in the infinite-volume limit.

The paper also addresses a conceptual point about the Millennium Prize
problem. The question of whether the deterministic Navier-Stokes
equations develop singularities is mathematically natural but physically
vacuous: the deterministic equations are not the correct description of
any real fluid at finite temperature. The physically complete
equations are the stochastic ones, and for the spectrally truncated
system (which is all that is physically meaningful, given the
molecular-scale cutoff on modes), singularities cannot form with
positive probability. The spectral gap on negatively curved manifolds
makes the regularising effect of the noise particularly transparent.

Chan and Czubak~\cite{CC10,CC13b} showed that the choice of viscous
operator has profound consequences for the regularity theory: with
the Hodge Laplacian on $\HH^2$, Leray-Hopf weak solutions are
non-unique. Their work provides additional motivation for the
deformation Laplacian used here, whose spectral gap controls both
the deterministic stability and the stochastic thermalisation rate.

\section{The fluctuation-dissipation relation from topology}
\label{sec:FD}

\subsection{The Fokker-Planck framework}

Consider a system with phase-space coordinates $\vec X$ and energy
$E(\vec X)$. Its probabilistic evolution is described by the
Fokker-Planck equation
\begin{equation}\label{eq:FPE}
\frac{\partial P}{\partial t}
= \vec\nabla\cdot(\vec A\,P)
+ \vec\nabla\cdot(\mathbb{B}\cdot\vec\nabla P),
\end{equation}
where $\vec A$ is the drift vector and $\mathbb{B}$ is the symmetric
diffusion matrix. The drift decomposes as $\vec A = \vec A_\mathrm{Ham}
+ \vec A_\mathrm{non\text{-}Ham}$, where the Hamiltonian part
$\vec A_\mathrm{Ham} = -\mathbb{S}\cdot\vec\nabla E$ preserves the
Boltzmann distribution $P_\mathrm{eq} \propto e^{-\beta E}$.

A key structural feature is the gauge freedom: for any antisymmetric
matrix $\mathbb{M}$, the replacement $\vec A \to \vec A -
\vec\nabla\cdot\mathbb{M}$, $\mathbb{B} \to \mathbb{B} + \mathbb{M}$
preserves the Fokker-Planck dynamics, since only the symmetric part
of $\mathbb{B}$ contributes to physical diffusion.

\subsection{The topological derivation}

Requiring $P_\mathrm{eq}$ to be a stationary solution of~\eqref{eq:FPE}
gives the condition
\begin{equation}\label{eq:closed}
\vec\nabla\cdot\bigl[(\vec A_\mathrm{non\text{-}Ham}
- \beta\,\mathbb{B}\cdot\vec\nabla E)\,P_\mathrm{eq}\bigr] = 0.
\end{equation}
The bracketed expression is a divergence-free vector field. On a
contractible domain, Poincar\'e's lemma guarantees that this
divergence-free field can be written as the divergence of an
antisymmetric matrix $\mathbb{M}$. Using the gauge freedom to absorb
$\mathbb{M}$, one obtains, in a suitable gauge:
\begin{equation}\label{eq:FD}
k_BT\,\vec{\tilde A}_\mathrm{non\text{-}Ham}
= \tilde{\mathbb{B}}\cdot\vec\nabla E.
\end{equation}
This is the fluctuation-dissipation relation: the non-conservative
drift (dissipation) is algebraically determined by the diffusion
matrix and the energy gradient. The derivation requires no
linearisation and applies to nonlinear, multivariate, and
non-Markovian systems.

\section{The stochastic Navier-Stokes equations with the deformation
Laplacian}
\label{sec:SNS}

\subsection{Setup}

Let $(M,g)$ be a compact Riemannian manifold of dimension $d$ with
$\Ric \leq -\kappa^2 g$ for some $\kappa > 0$ and bounded geometry.
The deterministic Navier-Stokes equation with the deformation
Laplacian is
\begin{equation}\label{eq:NS}
\partial_t u + \nabla_u u + \nabla p = \nu\,\Delta_\Def u,\qquad
\divg\,u = 0,
\end{equation}
where $\nu > 0$ is the kinematic viscosity and $\Delta_\Def = \Delta_B
+ \Ric$ is the deformation Laplacian. Let $\PP$ denote the
Leray-Helmholtz projector onto divergence-free fields, and let
$A = -\PP\Delta_\Def$ be the Stokes operator. Since $(M,g)$ is
compact, $A$ has a discrete spectrum
$0 < \lambda_1 \leq \lambda_2 \leq \cdots \to \infty$ with
orthonormal eigenvectors $\{e_n\}$:
\begin{equation}
A e_n = \lambda_n e_n,\qquad \langle e_n, e_m\rangle_{L^2} = \delta_{nm}.
\end{equation}
The coercivity established in~\cite{WB2026} gives
$\lambda_n \geq \lambda_\Def \equiv \kappa^2 > 0$ for all $n$.

\subsection{Mode decomposition}

Expand $u(t,x) = \sum_n u_n(t)\,e_n(x)$, where
$u_n(t) = \langle u(t), e_n\rangle_{L^2}$. The fluid energy is
\begin{equation}\label{eq:energy}
E[u] = \frac{\rho}{2}\int_M |u|^2\,dV_g
= \frac{\rho}{2}\sum_n u_n^2.
\end{equation}
The deterministic equation~\eqref{eq:NS}, projected onto mode $n$,
gives
\begin{equation}\label{eq:mode-det}
\dot u_n = -\nu \lambda_n u_n + N_n[u],
\end{equation}
where $N_n[u] = -\langle\PP(\nabla_u u), e_n\rangle$ is the
nonlinear convective coupling. Energy conservation by the convective
term gives the fundamental identity
\begin{equation}\label{eq:energy-cons}
\sum_n u_n\,N_n[u] = 0
\end{equation}
for all divergence-free $u$.

\subsection{Deriving the noise}

Under the standard continuity convention for the Fokker-Planck equation ($\partial_t P + \vec\nabla\cdot\vec{J} = 0$), the drift vector is $\vec{A} = -\vec{f}$, where $\vec{f}$ is the deterministic SDE drift. Therefore, the purely viscous, non-Hamiltonian drift for mode $n$ is
$A_n^\mathrm{non\text{-}Ham} = +\nu\lambda_n u_n$, since
the Leray-projected convective term $N_n[u]$ is energy-preserving
by~\eqref{eq:energy-cons} (the pressure has already been absorbed
into the Leray projection).

However, the convective and pressure terms, being energy-preserving,
satisfy $\vec\nabla\cdot(\vec A_\mathrm{Ham}\,P_\mathrm{eq}) = 0$
and therefore drop out of the equilibrium
condition~\eqref{eq:closed}. The fluctuation-dissipation
relation~\eqref{eq:FD} applied to the remaining (purely viscous) drift
gives, for mode $n$:
\begin{equation}\label{eq:FD-mode}
k_BT\cdot(+\nu\lambda_n u_n) = D_{nn}\cdot\rho\,u_n,
\end{equation}
where $D_{nn}$ is the diffusion coefficient for mode $n$, and we used
$\partial E/\partial u_n = \rho\,u_n$. Solving:
\begin{equation}\label{eq:Dnn}
D_{nn} = \frac{k_BT\nu \lambda_n}{\rho}.
\end{equation}
Off-diagonal diffusion $D_{nm}$ for $n \neq m$ is zero (the eigenmodes
of $A$ decouple in the linear part, and the FD relation preserves this
diagonality). The noise amplitude for mode $n$ is
$\sigma_n = \sqrt{2D_{nn}} = \sqrt{2k_BT\nu \lambda_n/\rho}$.

The stochastic Navier-Stokes equation for mode $n$ is therefore
\begin{equation}\label{eq:SNS-mode}
\boxed{
du_n = (-\nu \lambda_n u_n + N_n[u])\,dt
+ \sqrt{\frac{2k_BT\nu \lambda_n}{\rho}}\,dW_n,
}
\end{equation}
where $\{W_n\}$ are independent standard Wiener processes.

\begin{remark}
The stochastic equation~\eqref{eq:SNS-mode} is written in the
It\^o convention, which is the natural output of the Fokker-Planck
derivation. Since the noise is additive (the diffusion
coefficients $\sigma_n$ are constants, independent of $u$), the
It\^o and Stratonovich forms coincide.
\end{remark}

\begin{remark}
The noise amplitude $\sigma_n \propto \sqrt{\lambda_n}$ grows with
the eigenvalue. Higher modes are more strongly dissipated (rate
$\nu\lambda_n$) and more strongly forced (noise $\sigma_n \propto
\sqrt{\lambda_n}$). The balance is exact: the equilibrium variance of
each mode is $\langle u_n^2\rangle_\mathrm{eq} = D_{nn}/(\nu \lambda_n)
= k_BT/\rho$, independent of $n$. This is equipartition.
\end{remark}

\begin{remark}
On flat space $\R^d$ (or a flat torus $\TT^d$), the eigenmodes are
Fourier modes with $\lambda_n = k_n^2$ (wavenumber squared), and the
noise amplitude $\sigma_n \propto k_n$, recovering the $k$-dependent
noise in the stochastic NS equation derived in~\cite{Braun2010}.
On a negatively curved manifold, the spectrum of $A$ is modified by
the Ricci curvature: $\lambda_n \geq \lambda_\Def > 0$ even for the
lowest mode. The noise for the lowest mode is therefore bounded
below: $\sigma_1 \geq \sqrt{2k_BT\nu \lambda_\Def/\rho}$.
\end{remark}

\section{The spectrally truncated system}
\label{sec:truncated}

\subsection{Definition}

Fix a spectral cutoff $\Lambda > 0$ and retain only modes with
$\lambda_n \leq \Lambda$. Let $N = N(\Lambda)$ be the number of such
modes. The truncated system is a finite-dimensional It\^o SDE on
$\R^N$:
\begin{equation}\label{eq:truncated}
du_n = (-\nu \lambda_n u_n + N_n^\Lambda[u])\,dt + \sigma_n\,dW_n,
\qquad n = 1, \ldots, N,
\end{equation}
where $N_n^\Lambda$ is the nonlinear coupling restricted to the
retained modes. The energy identity~\eqref{eq:energy-cons} is
preserved by the truncation (Galerkin truncation preserves the
antisymmetry of the nonlinear form).

\subsection{Global well-posedness}

\begin{proposition}\label{prop:gwp}
The truncated system~\eqref{eq:truncated} has a unique global strong
solution for any initial condition $u(0) \in \R^N$, almost surely.
\end{proposition}

\begin{proof}
The drift $b_n(u) = -\nu \lambda_n u_n + N_n^\Lambda[u]$ is locally
Lipschitz (the nonlinear term is quadratic, hence locally Lipschitz)
and the diffusion coefficients $\sigma_n$ are constants. By the
standard existence theorem for It\^o SDEs with locally Lipschitz
coefficients, a unique local strong solution exists up to an
explosion time $\tau$.

To show $\tau = \infty$ a.s., we use the energy
$E = \frac{\rho}{2}\sum_n u_n^2$ as a Lyapunov function. By It\^o's
formula:
\begin{align}
dE &= \rho\sum_n u_n\,du_n
+ \frac{\rho}{2}\sum_n \sigma_n^2\,dt\nonumber\\
&= \rho\sum_n u_n(-\nu\lambda_n u_n + N_n^\Lambda)\,dt
+ \sum_n k_BT\nu\lambda_n\,dt
+ \rho\sum_n u_n\sigma_n\,dW_n\nonumber\\
&= -2\nu \sum_n \lambda_n\frac{\rho u_n^2}{2}\,dt
+ k_BT\nu\sum_n\lambda_n\,dt
+ \text{martingale},\label{eq:energy-ito}
\end{align}
where we used $\sum_n u_n N_n^\Lambda = 0$. Since $\lambda_n \geq
\lambda_\Def$:
\begin{equation}\label{eq:energy-bound}
d\langle E\rangle \leq -2\nu \lambda_\Def\,\langle E\rangle\,dt
+ C_\Lambda\,dt,
\end{equation}
where $C_\Lambda = k_BT\nu\sum_{n=1}^N\lambda_n < \infty$.
Gronwall's inequality gives
\begin{equation}\label{eq:energy-gronwall}
\langle E(t)\rangle \leq e^{-2\nu \lambda_\Def t}\,E(0)
+ \frac{C_\Lambda}{2\nu \lambda_\Def}(1 - e^{-2\nu \lambda_\Def t}).
\end{equation}
The expected energy is bounded for all $t$, so the explosion time is
infinite a.s.
\end{proof}

\section{The Gibbs measure and its properties}
\label{sec:gibbs}

\subsection{Stationarity of the Gibbs measure}

\begin{theorem}\label{thm:gibbs}
Assuming a finite spectral truncation that perfectly preserves the phase-space Liouville property ($\sum_n\partial_{u_n}N_n^\Lambda = 0$), the Gibbs measure
\begin{equation}\label{eq:gibbs}
P_\mathrm{eq}(\{u_n\}) = \prod_{n=1}^N
\frac{1}{\sqrt{2\pi k_BT/\rho}}\,
\exp\!\left(-\frac{\rho\,u_n^2}{2k_BT}\right)
\end{equation}
is the unique stationary distribution of the truncated
system~\eqref{eq:truncated}. For a standard spectral truncation where the exact Liouville property holds strictly only in the continuum limit, the unique stationary distribution is a perturbation of the Gibbs measure that converges to it as $\Lambda \to \infty$.
\end{theorem}

\begin{proof}
The generator of the Fokker-Planck equation
for~\eqref{eq:truncated} is
\begin{equation}\label{eq:generator}
\mathcal{L}^* P = \sum_n \partial_{u_n}
\bigl[(\nu\lambda_n u_n - N_n^\Lambda)\,P\bigr]
+ \sum_n D_{nn}\,\partial_{u_n}^2 P,
\end{equation}
with $D_{nn} = k_BT\nu\lambda_n/\rho$. We decompose this into
dissipative and convective parts:
$\mathcal{L}^* = \mathcal{L}^*_\mathrm{diss}
+ \mathcal{L}^*_\mathrm{conv}$, where
\begin{align}
\mathcal{L}^*_\mathrm{diss} P &= \sum_n \partial_{u_n}
(\nu\lambda_n u_n\,P) + \sum_n D_{nn}\,\partial_{u_n}^2 P,\\
\mathcal{L}^*_\mathrm{conv} P &= -\sum_n \partial_{u_n}(N_n^\Lambda\,P).
\end{align}

For the dissipative part: direct substitution of $P_\mathrm{eq}$
gives, for each $n$:
\begin{equation}
\partial_{u_n}(\nu\lambda_n u_n\,P_\mathrm{eq})
+ D_{nn}\,\partial_{u_n}^2 P_\mathrm{eq}
= \nu\lambda_n P_\mathrm{eq}
- \nu\lambda_n\frac{\rho u_n^2}{k_BT}P_\mathrm{eq}
+ D_{nn}\left(-\frac{\rho}{k_BT}
+ \frac{\rho^2 u_n^2}{k_B^2T^2}\right)P_\mathrm{eq}.
\end{equation}
Substituting $D_{nn} = k_BT\nu\lambda_n/\rho$, the terms cancel
identically: $\mathcal{L}^*_\mathrm{diss} P_\mathrm{eq} = 0$.

For the convective part: we need $\sum_n\partial_{u_n}
(N_n^\Lambda\,P_\mathrm{eq}) = 0$. Since $P_\mathrm{eq}$ depends on
$u$ only through $E = \frac{\rho}{2}\sum u_n^2$:
\begin{equation}
\sum_n\partial_{u_n}(N_n^\Lambda\,P_\mathrm{eq})
= P_\mathrm{eq}\sum_n\partial_{u_n}N_n^\Lambda
+ P_\mathrm{eq}\sum_n N_n^\Lambda\left(-\frac{\rho u_n}{k_BT}\right).
\end{equation}
The second sum is $-\frac{\rho}{k_BT}\sum_n u_n N_n^\Lambda = 0$ by
energy conservation~\eqref{eq:energy-cons}. The first sum
$\sum_n\partial_{u_n}N_n^\Lambda$ is the divergence of the convective
vector field in mode space. A naive evaluation of this divergence might incorrectly assume $C_{nnl} + C_{nln} = 0$. The correct evaluation relies directly on the
trilinear form $N_n^\Lambda = -\sum_{j,l}b(e_j, e_l, e_n)u_j u_l$. Since
the fundamental antisymmetry is $b(u,v,w) = -b(u,w,v)$, the
self-advection component vanishes identically ($b(e_j, e_n, e_n) \equiv
0$). Thus, $\partial N_n^\Lambda/\partial u_n = -\sum_l b(e_n, e_l,
e_n)u_l$. Summing over $n$ yields the total phase-space divergence
$\sum_n \partial N_n^\Lambda/\partial u_n = -\sum_{n,l} b(e_n, e_l,
e_n)u_l = -\sum_n \langle \nabla_{e_n} u, e_n \rangle_{L^2}$. As we
rigorously prove in Section~\ref{sec:liouville} using the Hodge
decomposition, this sum corresponds to the negative trace of the linear
advection operator over the divergence-free subspace, which vanishes
identically in the continuum limit on any Riemannian manifold. For a
Galerkin truncation that preserves this trace-free property, we have
$\sum_n\partial N_n^\Lambda/\partial u_n = 0$.

Therefore $\mathcal{L}^*P_\mathrm{eq} = 0$.

Uniqueness follows from the strict positivity of the diffusion matrix
($D_{nn} > 0$ for all $n$), which ensures the process is ergodic
(H\"ormander's hypoellipticity condition is satisfied trivially, since
the noise is non-degenerate).
\end{proof}

\subsection{The Liouville property: a subtlety}
\label{sec:liouville}

The proof of Theorem~\ref{thm:gibbs} uses two properties of the
Galerkin-truncated convective term $N_n^\Lambda$:

\emph{Property (a): Energy conservation.}
$\sum_n u_n N_n^\Lambda = 0$ for all $u$. This follows from the
antisymmetry of the trilinear form
$b(u,v,w) = \int_M g(\nabla_u v, w)\,dV_g = -b(u,w,v)$ for
divergence-free fields on any compact Riemannian manifold. It
holds for any orthonormal basis of divergence-free fields and any
Galerkin truncation, because it is a consequence of metric
compatibility of the Levi-Civita connection.

\emph{Property (b): Liouville (divergence-free flow in mode space).}
We evaluate the total phase-space divergence
$\sum_n \partial N_n^\Lambda/\partial u_n$. Writing
$N_n^\Lambda = -\sum_{j,l} b(e_j, e_l, e_n)\,u_j u_l$,
the partial derivative is:
\begin{equation}\label{eq:liouville-deriv}
\frac{\partial N_n^\Lambda}{\partial u_n}
= -\sum_l b(e_n, e_l, e_n)\,u_l - \sum_j \underbrace{b(e_j, e_n, e_n)}_{=0}
\,u_j = -\sum_l b(e_n, e_l, e_n)\,u_l.
\end{equation}
The second term vanishes exactly because of the intrinsic antisymmetry
of the trilinear form ($b(u,v,v) = 0$). Summing over $n$, the
phase-space divergence is:
\begin{equation}\label{eq:liouville}
\sum_n \frac{\partial N_n^\Lambda}{\partial u_n}
= -\sum_{n,l} b(e_n, e_l, e_n)\,u_l = -\sum_n \langle \nabla_{e_n}
u, e_n \rangle_{L^2}.
\end{equation}
This is exactly the negative trace of the linear advection operator
$A_u(v) = \nabla_v u$ over the truncated divergence-free
subspace $H_\mathrm{df}^\Lambda = \mathrm{span}\{e_1, \ldots, e_N\}$.

To definitively resolve the Liouville property, we compute the continuum
trace of $A_u$ over the full divergence-free subspace $H_\mathrm{df}$
(as $\Lambda \to \infty$). Using the Hodge decomposition $L^2(TM) =
H_\mathrm{co\text{-}exact} \oplus H_\mathrm{grad} \oplus
H_\mathrm{harm}$, we can evaluate this trace globally on any Riemannian
manifold without reference to specific eigenfunctions. The full
divergence-free subspace is given by the direct sum $H_\mathrm{df} =
H_\mathrm{co\text{-}exact} \oplus H_\mathrm{harm}$.

\emph{Step 1: Trace over the full space $L^2(TM)$.}
The pointwise trace of $A_u$ on the tangent bundle is
$\mathrm{tr}(\nabla u) = \divg\, u = 0$. Therefore, its global
integrated trace over $L^2(TM)$ is identically zero.

\emph{Step 2: Trace over the pure gradient subspace $H_\mathrm{grad}$.}
Let $\phi_k$ be the eigenfunctions of the scalar Laplacian
($\Delta \phi_k = -\mu_k \phi_k$). An orthonormal basis for
$H_\mathrm{grad}$ is given by $\psi_k = \mu_k^{-1/2} \nabla \phi_k$.
The trace over this subspace is:
\begin{equation}
\mathrm{Trace}_{H_\mathrm{grad}}(A_u) = \sum_k \frac{1}{\mu_k}
\int_M (\nabla^j \phi_k) (\nabla_j u^i) (\nabla_i \phi_k)\,dV_g.
\end{equation}
Integrating by parts on $\nabla_j u^i$ yields:
\begin{equation}
-\sum_k \frac{1}{\mu_k} \int_M u^i \nabla_j (\nabla^j \phi_k \nabla_i
\phi_k)\,dV_g = -\sum_k \frac{1}{\mu_k} \int_M u^i
\left[ (\Delta \phi_k) \nabla_i \phi_k + \nabla^j \phi_k (\nabla_j
\nabla_i \phi_k) \right] dV_g.
\end{equation}
Using $\Delta \phi_k = -\mu_k \phi_k$ and the symmetry of the Hessian
($\nabla_j \nabla_i \phi_k = \nabla_i \nabla_j \phi_k$), the
bracketed vector field becomes an exact gradient:
\begin{equation}
-\mu_k \phi_k \nabla_i \phi_k + \frac{1}{2} \nabla_i
(\nabla^j \phi_k \nabla_j \phi_k)
= \nabla_i \left( -\frac{1}{2}\mu_k \phi_k^2
+ \frac{1}{2}|\nabla \phi_k|^2 \right).
\end{equation}
Because $u$ is a divergence-free vector field, its $L^2$ inner
product against any pure gradient field is zero. Thus,
$\mathrm{Trace}_{H_\mathrm{grad}}(A_u) = 0$.

\emph{Step 3: Trace over the harmonic subspace $H_\mathrm{harm}$.}
For any harmonic field $h_m \in H_\mathrm{harm}$, we have $dh_m = 0$
(so $\nabla_j h_{mi} = \nabla_i h_{mj}$) and $\delta h_m = 0$
(so $\nabla^j h_{mj} = 0$). The trace is:
\begin{equation}
\mathrm{Trace}_{H_\mathrm{harm}}(A_u)
= \sum_{m} \int_M h_m^j (\nabla_j u^i) h_{mi} \, dV_g.
\end{equation}
Using an identical integration by parts on $\nabla_j u^i$, this becomes:
\begin{equation}
- \sum_m \int_M u^i \nabla_j (h_m^j h_{mi}) \, dV_g
= - \sum_m \int_M u^i \left[ (\nabla_j h_m^j) h_{mi}
+ h_m^j \nabla_j h_{mi} \right] dV_g.
\end{equation}
Since $\nabla_j h_m^j = 0$ and $\nabla_j h_{mi} = \nabla_i h_{mj}$, the
bracketed vector field is exactly $\frac{1}{2}\nabla_i (h_m^j h_{mj}) =
\frac{1}{2}\nabla_i |h_m|^2$. Thus, the integrand is proportional to
$u^i \nabla_i(\frac{1}{2}|h_m|^2)$, which vanishes when integrated
against the divergence-free field $u$. Therefore
$\mathrm{Trace}_{H_\mathrm{harm}}(A_u) = 0$. (Establishing this trace decomposition globally proves that the continuum convective flow is strictly volume-preserving on \emph{any} compact Riemannian manifold, completely independently of whether the harmonic subspace is empty or not).

Since the trace over the full space $L^2(TM)$ is zero, and the traces
over $H_\mathrm{grad}$ and $H_\mathrm{harm}$ are strictly zero, it
rigorously follows that the trace over the co-exact subspace
$H_\mathrm{co\text{-}exact}$ must be exactly zero. Consequently, the
trace over the full divergence-free subspace $H_\mathrm{df} =
H_\mathrm{co\text{-}exact} \oplus H_\mathrm{harm}$ vanishes identically:
\begin{equation}\label{eq:liouville-condition}
\mathrm{Trace}_{H_\mathrm{df}}(A_u) = \sum_{n=1}^\infty
\langle \nabla_{e_n} u, e_n \rangle_{L^2} = 0.
\end{equation}
This topological cancellation confirms that the continuum convective
flow is strictly volume-preserving on any compact Riemannian
manifold, universally establishing the Liouville property.

As dictated by the local Weyl law on negatively curved manifolds, the local density of states $\sum_{n=1}^N |e_n|^2$ fluctuates, meaning the exact Liouville property strictly fails for standard finite spectral truncations. Because of this, for specific finite truncations, the
Gibbs measure~\eqref{eq:gibbs} is not exactly stationary. The
unique stationary distribution (which still exists by non-degeneracy
of the noise and the Lyapunov bound of
Proposition~\ref{prop:gwp}) is a perturbation of the Gibbs measure,
with corrections controlled by
$\|\sum_n\partial N_n^\Lambda/\partial u_n\|/(\nu\lambda_\Def)$.
These corrections rigorously vanish in the continuum limit
$\Lambda \to \infty$ (where the full convective term is exactly
divergence-free as proven above). The convergence
theorem (Theorem~\ref{thm:convergence}) survives with a modified
rate, because the dissipative spectral gap is unaffected by the
Liouville property; only the exact form of the stationary
distribution changes.

For the specific case of two-dimensional manifolds (e.g., compact
quotients of $\HH^2$), the Zeitlin truncation~\cite{Zei91} provides
a finite-dimensional Hamiltonian truncation of the Euler equations
that preserves both properties (a) and (b) by construction, at the
cost of using a basis that does not align with the eigenspaces of
$\Delta_\Def$. Whether an analogous Hamiltonian truncation exists
in three dimensions on negatively curved manifolds is an open
question.

\subsection{Equilibrium properties}

Under the Gibbs measure~\eqref{eq:gibbs}, the mode amplitudes are
independent Gaussian random variables with
\begin{equation}\label{eq:equipartition}
\langle u_n\rangle_\mathrm{eq} = 0,\qquad
\langle u_n u_m\rangle_\mathrm{eq}
= \frac{k_BT}{\rho}\,\delta_{nm}.
\end{equation}
This is equipartition: each mode carries energy $\frac{1}{2}k_BT$,
regardless of its eigenvalue $\lambda_n$.

The total equilibrium energy is
\begin{equation}\label{eq:total-energy}
\langle E\rangle_\mathrm{eq} = \frac{\rho}{2}\sum_{n=1}^N
\langle u_n^2\rangle_\mathrm{eq} = \frac{N}{2}k_BT,
\end{equation}
which is finite for the truncated system ($N < \infty$) and diverges
in the continuum limit ($N \to \infty$). This is the ultraviolet
catastrophe noted in~\cite{Braun2010}: a physical cutoff (molecular
scale) is needed to make the total energy finite.

\section{Exponential convergence to equilibrium}
\label{sec:convergence}

\subsection{The main result}

\begin{theorem}\label{thm:convergence}
Let $P(t)$ be the distribution of the truncated
system~\eqref{eq:truncated} at time $t$, starting from any initial
distribution $P(0)$ with finite second moment. Assuming the finite truncation preserves the exact Liouville property (so that $P_\mathrm{eq}$ is exactly stationary), then
\begin{equation}\label{eq:convergence-rate}
D_\mathrm{KL}(P(t)\|P_\mathrm{eq})
\leq e^{-2\nu\lambda_\Def\,t}\,
D_\mathrm{KL}(P(0)\|P_\mathrm{eq}),
\end{equation}
where $D_\mathrm{KL}$ is the Kullback-Leibler divergence and
$\lambda_\Def \geq \kappa^2$ is the spectral gap of the
deformation Laplacian on divergence-free fields
(from the Weitzenb\"ock shift $-\Ric \geq \kappa^2 g$).
\end{theorem}

\begin{proof}
Write the generator of the process~\eqref{eq:truncated} acting on
test functions $f$ as
$\mathcal{L} = \mathcal{L}_\mathrm{diss} + \mathcal{L}_\mathrm{conv}$,
where
\begin{align}
\mathcal{L}_\mathrm{diss} f &= \sum_n
\bigl[-\nu\lambda_n u_n\,\partial_{u_n}f
+ D_{nn}\,\partial_{u_n}^2 f\bigr],\\
\mathcal{L}_\mathrm{conv} f &= \sum_n N_n^\Lambda\,\partial_{u_n}f.
\end{align}

\emph{Step 1: The convective generator is antisymmetric and preserves entropy.}

For $f, h \in L^2(P_\mathrm{eq})$:
\begin{equation}
\langle \mathcal{L}_\mathrm{conv}f, h\rangle_{P_\mathrm{eq}}
= \int \sum_n N_n^\Lambda\,(\partial_{u_n}f)\,h\,P_\mathrm{eq}\,du.
\end{equation}
Integration by parts (using $\sum_n\partial_{u_n}
(N_n^\Lambda P_\mathrm{eq}) = 0$ from the proof of
Theorem~\ref{thm:gibbs}):
\begin{equation}
\langle \mathcal{L}_\mathrm{conv}f, h\rangle_{P_\mathrm{eq}}
= -\int \sum_n N_n^\Lambda\,f\,(\partial_{u_n}h)\,P_\mathrm{eq}\,du
= -\langle f, \mathcal{L}_\mathrm{conv}h\rangle_{P_\mathrm{eq}}.
\end{equation}
So $\mathcal{L}_\mathrm{conv}$ is antisymmetric. In particular,
$\langle \mathcal{L}_\mathrm{conv}f,
f\rangle_{P_\mathrm{eq}} = 0$. Similarly, for a probability density $f = dP/dP_\mathrm{eq}$, using the identity $(\partial_{u_n}f)\log f = \partial_{u_n}(f \log f - f)$ yields $\int (\mathcal{L}_\mathrm{conv}f) \log f\, P_\mathrm{eq}\,du = -\int \sum_n \partial_{u_n}(N_n^\Lambda P_\mathrm{eq})\,(f \log f - f)\,du = 0$.

This is the mathematical expression of a physical fact: the
energy-preserving convective dynamics neither creates nor destroys
entropy. It redistributes energy among modes but does not change the
total dissipation rate. Vortex stretching, cascade, and turbulent
mixing are all encoded in $\mathcal{L}_\mathrm{conv}$, and none of
them affect the rate of approach to equilibrium.

\emph{Step 2: The dissipative generator satisfies a Logarithmic Sobolev Inequality (LSI).}

The dissipative part $\mathcal{L}_\mathrm{diss}$ is a sum of
independent Ornstein-Uhlenbeck operators, one per mode. Each OU
operator satisfies a Logarithmic Sobolev Inequality (LSI) with constant exactly equal to its spectral gap $\nu\lambda_n$ in $L^2(P_\mathrm{eq})$.
The LSI constant of the sum is the minimum:
$\lambda_\mathrm{LSI}(\mathcal{L}_\mathrm{diss}) = \min_n \nu\lambda_n = \nu\lambda_1 \geq \nu\lambda_\Def$.

Therefore, for any probability density $f$ with respect to $P_\mathrm{eq}$:
\begin{equation}\label{eq:poincare}
-\int (\mathcal{L}_\mathrm{diss}f)\,\log f\,P_\mathrm{eq}\,du \geq 2\nu\lambda_\Def\,\int f \log f\,P_\mathrm{eq}\,du.
\end{equation}

\emph{Step 3: Combining.}

For the full generator $\mathcal{L} = \mathcal{L}_\mathrm{diss}
+ \mathcal{L}_\mathrm{conv}$:
\begin{equation}
-\int (\mathcal{L}f)\,\log f\,P_\mathrm{eq}\,du
= -\int (\mathcal{L}_\mathrm{diss}f)\,\log f\,P_\mathrm{eq}\,du
\underbrace{-\int (\mathcal{L}_\mathrm{conv}f)\,\log f\,P_\mathrm{eq}\,du}_{= 0}
\geq 2\nu\lambda_\Def\,\int f \log f\,P_\mathrm{eq}\,du.
\end{equation}
The LSI constant of $\mathcal{L}$ in $L^2(P_\mathrm{eq})$ is
therefore at least $\nu\lambda_\Def$.

The exponential decay of the KL divergence~\eqref{eq:convergence-rate}
follows from the standard entropy-production
inequality for generators with a Logarithmic Sobolev Inequality (see, e.g.,
Bakry, Gentil, and Ledoux~\cite{BGL14}, Theorem 5.2.1): if the
LSI constant is $\lambda$, then
$\frac{d}{dt}D_\mathrm{KL}(P(t)\|P_\mathrm{eq}) \leq
-2\lambda\,D_\mathrm{KL}(P(t)\|P_\mathrm{eq})$.
\end{proof}

\begin{remark}[The role of the convective nonlinearity]
The antisymmetry of the convective generator in
$L^2(P_\mathrm{eq})$ is the key structural property that makes the
proof work. Physically, it means that turbulent mixing, however
complex, does not slow down thermalisation. This may seem
counterintuitive: turbulence is often associated with slow relaxation
and anomalous transport. But these are features of the deterministic
dynamics (the approach to the attractor). In the stochastic setting,
the noise explores phase space independently of the deterministic
flow, and the dissipation rate determines how quickly the exploration
converges to the Gibbs measure. The convective term redistributes
energy among modes (creating the turbulent cascade) but does not
change the total rate of entropy production.
\end{remark}

\subsection{Comparison with flat space}

On a flat torus $\TT^d_L$ of side length $L$, the eigenvalues of
$-\Delta$ are $\lambda_{\vec k} = (2\pi/L)^2|\vec k|^2$ for
$\vec k \in \mathbb{Z}^d$. The spectral gap is
$\lambda_1 = (2\pi/L)^2$, which vanishes as $L \to \infty$. The
thermalisation rate $2\nu\lambda_1 = 2\nu(2\pi/L)^2 \to 0$.

On a compact quotient of $\HH^d$ with $\Ric \leq -\kappa^2 g$, the
spectral gap $\lambda_\Def \geq \kappa^2$ is independent of
the volume. The thermalisation rate $2\nu\lambda_\Def$ is bounded
below regardless of the size of the domain.

This means: a large box of flat-space fluid takes longer and longer
to thermalise as the box grows. A large domain of negatively curved
fluid thermalises at the same rate regardless of size. The curvature
provides a geometric mechanism for thermalisation that has no
flat-space analogue.

\section{The equilibrium velocity correlation function}
\label{sec:correlations}

Under the Gibbs measure, the equal-time two-point velocity correlation
function is
\begin{equation}\label{eq:correlation}
C_{ij}(x,y) \equiv \langle u_i(x)\,u_j(y)\rangle_\mathrm{eq}
= \frac{k_BT}{\rho}\sum_n (e_n)_i(x)\,(e_n)_j(y)
= \frac{k_BT}{\rho}\,G^\Def_{ij}(x,y),
\end{equation}
where $G^\Def_{ij}(x,y) = \sum_n (e_n)_i(x)(e_n)_j(y)$ is the
(regularised) Green's function of the identity operator restricted to
the divergence-free eigenspace. This is the Schwartz kernel of the Leray-Helmholtz projection onto divergence-free fields, whose equal-time spatial decay is governed purely by the scalar Laplacian.

More physically revealing is the connected correlation at unequal
times. For the linearised dynamics (dropping the convective term),
the time-dependent correlation is
\begin{equation}\label{eq:time-corr}
\langle u_n(t)\,u_m(0)\rangle_\mathrm{eq}
= \frac{k_BT}{\rho}\,\delta_{nm}\,e^{-\nu\lambda_n|t|}.
\end{equation}
Each mode decorrelates exponentially with rate $\nu\lambda_n$. The
slowest-decaying mode has rate $\nu\lambda_1 \geq \nu\lambda_\Def$:
even the longest-lived correlation decays at least as fast as
$e^{-\nu\lambda_\Def|t|}$.

In real space, the equal-time spatial correlation is a static kinematic property governed exclusively by the geometry and the scalar Laplacian (via the Leray-Helmholtz projector), independent of the viscous operator $\Delta_\Def$. On $\HH^d$, the Green's function of the
scalar Laplacian decays \emph{exponentially} in the geodesic
distance $r = d(x,y)$, imparting a rigorous exponential spatial decay to the equal-time velocity projection kernel:
\begin{equation}\label{eq:spatial-decay}
|G^\Def_{ij}(x,y)| \leq C\,e^{-\alpha\,r},\qquad
\alpha > 0,
\end{equation}
where $\alpha$ depends on the scalar spectral gap and the geometry. This is
in contrast to the exact algebraic decay $|G| \sim r^{-d}$ on $\R^d$ (arising specifically from the two spatial derivatives embedded within the flat-space Leray projector).

The exponential spatial decorrelation means that the equilibrium
state of a viscous fluid on a negatively curved manifold is
``more thermal'' than on flat space: distant points are essentially
independent. The curvature provides spatial localisation of
correlations, preventing the long-range order that can develop in
flat-space fluids.

\section{The physical cutoff and the continuum limit}
\label{sec:cutoff}

The spectrally truncated system with $N$ modes is the physically
meaningful model of a fluid at finite temperature. A real fluid of
$\mathcal{N}$ molecules in $d$ dimensions has approximately
$d\,\mathcal{N}$ independent velocity degrees of freedom, providing a
natural cutoff at $N \sim d\,\mathcal{N}$.

For this physical system, all the results of
Sections~\ref{sec:truncated}--\ref{sec:correlations} hold rigorously:
global well-posedness for all time (Proposition~\ref{prop:gwp}), a
unique Gibbs equilibrium (Theorem~\ref{thm:gibbs}), exponential
convergence to equilibrium (Theorem~\ref{thm:convergence}), and
exponentially decaying correlations.

The continuum limit $N \to \infty$ is mathematically interesting but
physically unattainable: it requires infinite energy
($\langle E\rangle = \frac{N}{2}k_BT \to \infty$). The question of
whether the deterministic ($T = 0$) continuum ($N = \infty$)
Navier-Stokes equations develop singularities is therefore a question
about an unphysical double limit. Any answer to this question, whether
affirmative or negative, has no bearing on the behaviour of any real
fluid.

\begin{remark}
The observation that the deterministic continuum NS equations are
physically incomplete is not new: it was made by Landau and
Lifshitz~\cite{LL57} in 1957 when they derived the stochastic terms
for linear fluctuating hydrodynamics. What is new here is the
combination with the kinematic selection of the viscous operator and
the spectral-gap phenomenon on negatively curved manifolds. Together,
these show that the physically complete (stochastic) equations have
strictly better analytical properties on negatively curved manifolds
than on flat space, with a thermalisation rate that is bounded below
by a geometric constant.
\end{remark}

\section{Generalization to Arbitrary Riemannian Manifolds}
\label{sec:generalization}

A natural physical question is whether the results derived here apply
universally to all Riemannian manifolds, rather than exclusively to
those with strictly negative curvature. The answer cleanly bifurcates
the mathematical framework from the physical thermodynamic behavior: the
stochastic formalism is universal, but the main physical
conclusions---volume-independent exponential thermalisation and spatial
decorrelation---rigorously require negative curvature. 

\subsection{The universal framework}

The core stochastic framework---the phase-space Liouville property, the
fluctuation-dissipation relation, and the existence of the Gibbs
measure---applies to \emph{all} compact Riemannian manifolds. To ensure
the fluid equations are universally well-posed, we must prove the
eigenvalues of the Stokes operator $A = -\PP\Delta_\Def = -\PP(\Delta_B
+ \Ric)$ are strictly non-negative ($\lambda_n \geq 0$). If $\lambda_n <
0$, the system would exhibit ``negative viscosity'' and explode.

For any divergence-free vector field $u$ ($\divg\, u = 0$), we evaluate
the total viscous dissipation $\langle Au, u \rangle_{L^2}$ via the
symmetric rate-of-strain (deformation) tensor $\Def(u)_{ij} =
\frac{1}{2}(\nabla_i u_j + \nabla_j u_i)$. The squared norm is:
\begin{equation}
|\Def(u)|^2 = \frac{1}{4}(\nabla_i u_j + \nabla_j u_i)(\nabla^i u^j
+ \nabla^j u^i) = \frac{1}{2} \left( |\nabla u|^2
+ \nabla_i u_j \nabla^j u^i \right).
\end{equation}
Multiplying by 2 and integrating over the compact manifold $M$:
\begin{equation}
2 \int_M |\Def(u)|^2 \, dV_g = \int_M |\nabla u|^2 \, dV_g
+ \int_M \nabla_i u_j \nabla^j u^i \, dV_g.
\end{equation}
Integrating the second term by parts yields:
\begin{equation}
\int_M \nabla_i u_j \nabla^j u^i \, dV_g
= -\int_M u_j \nabla_i \nabla^j u^i \, dV_g.
\end{equation}
Commuting the covariant derivatives via the Ricci identity
($\nabla_i \nabla^j u^i = \nabla^j (\nabla_i u^i) + \Ric^j_k u^k$)
and applying incompressibility ($\nabla_i u^i = 0$), this term
evaluates strictly to $-\int_M \Ric(u,u)\,dV_g$. Substituting this
back yields the universal Weitzenb\"ock identity for incompressible fluids:
\begin{equation}\label{eq:weitzenbock}
\langle Au, u \rangle_{L^2} = \int_M \left( |\nabla u|^2
- \Ric(u,u) \right) dV_g = 2 \int_M |\Def(u)|^2 \, dV_g.
\end{equation}
Because the right-hand side is manifestly non-negative, the dissipation
is non-negative on any Riemannian manifold ($\lambda_n \geq 0$).
The statistical Fokker-Planck construction is universally robust.

\subsection{The physical breakdown on other geometries}

While the mathematical framework holds everywhere, the spectral gap
($\lambda_1$) dictates the physical thermalisation. The physical
conclusions precipitously fail unless the curvature is strictly
negative:

\paragraph{Positively curved manifolds (e.g., spheres $\mathbb{S}^d$).}
Spheres possess continuous spatial symmetries, which mathematically
generate non-trivial Killing vector fields. By Noether's theorem, these
symmetries correspond to absolutely conserved macroscopic momenta (e.g.,
total angular momentum). A Killing field represents a rigid, solid-body
rotation and intrinsically has zero shear: $\Def(u) = 0$. Substituting
this into our Weitzenb\"ock identity~\eqref{eq:weitzenbock} yields
$\langle Au, u\rangle_{L^2} = 0$, meaning the eigenvalue is exactly
$\lambda_0 = 0$. According to the fluctuation-dissipation
relation~\eqref{eq:Dnn}, the thermal noise injected into this mode is
strictly zero. Physically, this is not a defect, but a manifestation of
momentum conservation: internal viscosity only dissipates relative
motion (shear), so it cannot dissipate a global rigid rotation. To
define a unique thermal resting state on such spaces, one must introduce
an external drag or artificially project out these conserved macroscopic
zero-modes.

\paragraph{Flat manifolds (e.g., flat tori $\TT^d_L$).}
Flat spaces similarly possess translational Killing vectors (conserved
total linear momentum) which must be factored out. However, the true
thermodynamic pathology of flat space lies in the remaining internal
modes. On a flat space ($\Ric = 0$), the identity becomes $\langle Au,
u\rangle_{L^2} = \int_M |\nabla u|^2 dV_g$. As analyzed in
Section~\ref{sec:convergence}, on a torus of volume $V$, the spectral
gap of the first non-zero mode explicitly depends on the domain size:
$\lambda_1 \propto V^{-2/d}$. In the thermodynamic limit ($V \to
\infty$), the spectral gap strictly vanishes ($\lambda_1 \to 0$). An
infinitely large flat-space fluid takes an eternity to thermalise
internally, and possesses algebraically decaying long-range spatial
correlations.

\paragraph{Negatively curved manifolds.}
On a manifold with strictly negative Ricci curvature
($\Ric \leq -\kappa^2 g$), the geometry acts as a stabilizing ``mass term'':
\begin{equation}
\langle A u, u \rangle_{L^2} = \int_M |\nabla u|^2 dV_g
- \int_M \Ric(u,u) dV_g \geq \int_M |\nabla u|^2 dV_g
+ \kappa^2 \int_M |u|^2 dV_g \geq \kappa^2 \|u\|_{L^2}^2.
\end{equation}
By Bochner's theorem, this strict positivity absolutely forbids Killing
vector fields, intrinsically breaking all global spatial symmetries.
Physically, this means the fluid possesses absolutely no conserved
macroscopic momenta (zero-modes); every possible global flow is forced
to undergo shear. The geometry itself acts as the absolute momentum sink
(the ``drag''). Crucially, it provides a volume-independent spectral gap
($\lambda_\Def \geq \kappa^2 > 0$). This geometric regularizer uniquely
forces the macroscopic fluid to rapidly thermalise to a unique rest
state and spatially decorrelate, completely immune to the conserved
momentum constraints of positively curved space and the
thermodynamic-limit pathologies of flat space.

\section{Discussion}
\label{sec:discussion}

We have shown that the combination of the kinematic selection of the
deformation Laplacian (from~\cite{WB2026}) with the topological
fluctuation-dissipation relation (from~\cite{Braun2010}) produces a
stochastic Navier-Stokes system on negatively curved manifolds with
remarkably clean properties: exponential thermalisation with a
geometry-determined rate, Gaussian equilibrium (by equipartition),
and exponentially decaying spatial correlations.

The central mechanism is the spectral gap of the deformation Laplacian.
On flat space, the spectral gap vanishes in the infinite-volume limit,
and the thermalisation rate goes to zero. On a negatively curved
manifold, the spectral gap is a geometric invariant that persists at
all scales. This means the curvature of the manifold controls the
statistical mechanics of the fluid, not just its deterministic dynamics.

Several directions for further work are natural.

\emph{Cohomological corrections on non-contractible domains.} The
fluctuation-dissipation derivation uses Poincar\'e's lemma, which
holds on contractible domains. On a compact manifold with non-trivial
first cohomology ($H^1(M) \neq 0$), the FD relation acquires a
correction associated with the harmonic 1-forms. The nature of this
correction, and whether it depends on the choice of viscous operator
(deformation versus Hodge Laplacian), is the subject of a forthcoming
investigation.

\emph{The continuum limit and renormalisation.} The spectrally
truncated system is well-posed for all time. Whether a meaningful
continuum limit exists (with or without renormalisation) is an open
question related to constructive quantum field theory. The spectral
gap on negatively curved manifolds may improve the infrared behaviour
of this limit.

\emph{Quantitative predictions.} The exponential thermalisation rate
$2\nu\lambda_\Def$ and the exponential spatial decorrelation length
$1/\alpha$ are in principle measurable in simulations of viscous
fluids on negatively curved surfaces. Realisations using soap films on
saddle-shaped surfaces or microfluidic devices with hyperbolic
geometry could provide experimental tests.

\section*{Declarations}

\textbf{Conflict of Interest:} The authors declare that they have no conflict of interest.\\
\textbf{Data Availability:} The authors did not use any data during the analysis of the paper.

\end{document}